\documentclass[runningheads,oribibl]{llncs}

\usepackage{graphicx}
\usepackage[dvipsnames,usenames]{color}
\usepackage{amssymb, amsmath}
\usepackage{array}     	
\usepackage[ruled, vlined]{algorithm2e}
\usepackage{verbatim}  	
\usepackage{tabularx}  	
\usepackage{multirow}  	
\usepackage{booktabs} 	
\usepackage[table]{xcolor}  
\usepackage{subfig}  		

\usepackage[bookmarks]{hyperref}
\usepackage{url}

\newcommand{\remove}[1]{}

\newcommand{\dist}{{\rm dist}}

\newcommand{\depth}{{\rm depth}}
\newcommand{\radius}{{\rm rad}}
\newcommand{\col}{{\rm color}}
\newcommand{\ww}{{\mathfrak{w}}}

\renewcommand{\multirowsetup}{\centering}

\begin{document}

\mainmatter  

\title{Oblivious Buy-at-Bulk in Planar Graphs\thanks{This work is supported by NSF grant CNS-084608.}
}
\titlerunning{Oblivious Buy-at-Bulk in Planar Graphs}

\author{Srivathsan Srinivasagopalan \and Costas Busch \and S.S. Iyengar}
\institute{Computer Science Department \\
						Louisiana State University \\
						\email{\{ssrini1, busch, iyengar\}@csc.lsu.edu}}

\authorrunning{Srivathsan Srinivasagopalan, Costas Busch, S.S. Iyengar}

\date{}
\maketitle

\begin{abstract}
In the oblivious buy-at-bulk network design problem in a graph,
the task is to compute a fixed set of paths for every pair of source-destinations in the graph,
such that any set of demands can be routed along these paths. 
The demands could be aggregated at intermediate edges 
where the fusion-cost is specified by a canonical (non-negative concave) function $f$.
We give a novel algorithm for planar graphs which is oblivious 
with respect to the demands, and is also oblivious with respect to the fusion function $f$. 
The algorithm is deterministic and computes 
the fixed set of paths in polynomial time,
and guarantees a $O(\log n)$ approximation ratio for any set of demands and any
canonical fusion function $f$,
where $n$ is the number of nodes. 
The algorithm is asymptotically optimal, since it is known that this problem
cannot be approximated with better than $\Omega(\log n)$ ratio.
To our knowledge,
this is the first tight analysis for planar graphs,
and improves the approximation ratio by a factor of $\log n$ with respect 
to previously known results.	
\end{abstract}

\section{Introduction}
\label{section:Introduction}

A typical client-server model has many clients and multiple servers where a subset of the client set wishes to route a certain amount of data to a subset of the servers at any given time. The set of clients and the servers are assumed to be geographically far apart. To enable communication among them, there needs to be a network of cables deployed. Moreover, the deployment of network cables has to be of minimum cost that also minimizes the communication cost among the various network components. This is what we roughly call as a typical network design problem. The same problem can be easily applied to many similar practical scenarios such as oil/gas pipelines and the Internet.

The ``Buy-at-Bulk'' network design considers the economies of scale into account. As observed in \cite{1170547}, in a telecommunication network, bandwidth on a link can be purchased in some discrete units $u_1 < u_2 < \dots < u_j$ with costs $c_1 < c_2 < \dots < c_j$ respectively. The economies of scale exhibits the property where the cost per bandwidth decreases as the number of units purchased increases: $c_1/u_1 > c_2/u_2 > \dots c_j/u_j$. This property is the reason why network capacity is bought/sold in ``wholesale'', or why vendors provide ``volume discount''.

The generalized form of the buy-at-bulk problem is where there are multiple demands 
from sources to destinations, and it is commonly referred as \emph{Multi-Sink Buy-at-Bulk} (MSBB). Typically, the demand flows are in discrete units and are unsplittable (indivisible), i.e., the flow follows a single path from the demand node to its destination. These problems are often called ``discrete cost network optimization'' in operations research.

As mentioned in \cite{644191}, if information flows from $x$ different sources over a link, then, the cost of total information that is transmitted over that link is proportional to $f(x)$, where $f: \mathbb{Z}^+ \rightarrow \mathbb{R}^+$. The function $f$ is called a \textit{canonical} fusion function if it is concave, non-decreasing, $f(0)=0$ and has the subadditive property $f(x_1 + x_2) \leq f(x_1) + f(x_2)$, $\forall x_1,x_2, (x_1+x_2) \in \mathbb{Z}^+$. Generally, MSBB problems use the subadditive property to ensure that the `size' of the aggregated data is smaller than the sum of the sizes of individual data.

We study the \emph{oblivious} buy-at-bulk network design problem (MSBB) with the following constraints: an unknown set of demands and an unknown concave fusion cost function $f$. An abstraction of this problem can be found in many applications, one of which is data fusion in wireless sensor networks where data from sensors is aggregated over time in multiple sinks.
Other application include Transportation \& Logistics (railroad, water, oil, gas pipeline construction) etc.
Many of these problems are formulated as networks on a plane that can be mapped to planar graphs.

\subsection{Problem Statement}

Assume that we are given a weighted graph $G = (V, E, \ww)$, with edge weights $\ww: E \longrightarrow \mathbb{Z}^{+}$. We denote $\ww_e$ to be the weight of edge $e$. Let $d_i = (s_i, t_i)$ be a unit of demand that induces an unsplittable unit of flow from source node $s_i \in V$ to destination node $t_i \in V$. Let $A = \{d_1, d_2, \ldots, d_r\}$ be a set of demands that are routed through paths in $G$. It is possible that some paths may overlap. The flow of these demands forms a set of paths $P(A) = \{p(d_1), p(d_2), \ldots, p(d_r)\}$. 

There is an arbitrary canonical function $f$ at every edge where data aggregates. 
This $f$ is same for all the edges in $G$. 
Let $\varphi_e(A)= \{p(d_i): e \in p(d_i)\}$ 
denote the set of paths that use an edge $e \in E$. 
Then, we define the cost of an edge $e$ to be $C_e(A) = f(|\varphi_e(A)|) \cdot w_e$. 
The total cost of the set of paths is defined to be $C(A) = \sum_{e \in E}{C_e(A)}$.
For this set $A$, there is an optimal set of paths $P^*(A)$ with respective cost $C^*(A)$. 
The approximation ratio for the paths $P(A)$ is defined as $\frac{C(A)}{C^*(A)}$.
The MSBB optimization problem on input $A$ is to find a set of paths $P(A)$ 
that minimizes the approximation ratio.
We note that MSBB is NP-Hard as the Steiner tree problem is its special case (when $f(x) = 1$ and when there is only one destination node) \cite{589033}.

An {\em oblivious algorithm} ${\cal A}_{obl}$ for the MSBB problem,
computes a fixed set of paths, denoted $P(G)$ 
for every pair of source destination nodes in $V$.
Given any set of demands $A$, the path $p(d_i)$ for each $d_i = (s_i, t_i) \in A$,
is the fixed path in $P(G)$ from $s_i$ to $t_i$.
This gives a set of paths $P(A)$ to route the demands $A$.
We define the {\em approximation ratio} of ${\cal A}_{obl}$, as: 
\[
    A.R.({\cal{A}}_{obl}) = \max_{A}\frac{C(A)}{C^*(A)}.
\]
We aim to find algorithms that minimizes the above approximation ratio
for any canonical function $f$ which is unknown to the algorithm.
The best known oblivious algorithm is by by Gupta {\it et al.} \cite{1109665}
and provides approximation ratio $O(\log^2 n)$ for general graphs.
No better result is known for planar graphs. 
This problem is NP-hard, since MSBB is NP-hard.

\subsection{Contribution}
We provide an oblivious algorithm {\sf FindPaths} for MSBB problems in planar graphs.
Our algorithm is deterministic and computes in polynomial time 
a fixed set of paths that guarantees $O(\log n)$-approximation ratio 
for any canonical function $f$ (where $f$ is unknown to the algorithm).
We also give a lower bound for the approximation ratio for $A$ to be of $\Omega(\log n)$,
where $n$ is the number of nodes in the graph.
A lower bound of $\Omega(\log n)$ for planar graphs is provided in the context 
of the online Steiner tree problem by Imase and Waxman \cite{imase}. 
Thus, our bound is tight with respect to planar graphs.
It is also a $\log n$ factor improvement over the best previously known result \cite{1109665}.


We build the set of paths based on sparse covers
(see \cite{355459} for an overview of sparse covers).
A $\gamma$-cover consists of clusters where for each node there is some cluster 
that contains its $\gamma$-neighborhood.
We construct $O(\log n)$ levels of covers with exponentially increasing 
locality parameter $\gamma$.
For every cluster we elect a leader.
For any pair of nodes $u, v$ we identify an appropriate common lowest-level cluster
that contains both $u$ and $v$,
and the cluster has a respective common leader $\ell$.
Then the path from $u$ to $v$ 
is formed by connecting successive path segments emanating from both $u$ and $v$ 
and using intermediate leaders of lower level clusters
until the common leader $\ell$ is reached.

In the analysis, we introduce the notion of coloring sparse covers,
where two clusters that are close receive different color.
We show the existence of a sparse cover with constant coloring 
(based on the sparse covers in \cite{1281112}).
This enables us to obtain optimal approximation at every level.
When we combine all the levels, we get an $O(\log n)$ approximation.

\subsection{Related Work}

\subsubsection{Oblivious Network Design}
Below, we present the related work on oblivious network design and Table \ref{table:comparison} summarizes some results and compares our work with their's. What distinguishes our work with the others' is the fact that we provide a set of paths for the MSBB problem while others provide an overlay tree for SSBB version.

\renewcommand{\multirowsetup}{\centering}
\newcommand {\otoprule}{\midrule [\heavyrulewidth]}  

\begin{table}[t]
 		\caption{Our results and comparison with previous results for data-fusion schemes. $n$ is the total number of nodes in the topology, $k$ is the total number of source nodes.}
  \footnotesize																										 
 	\begin{tabular}{*{3}{m{2.2cm}}*{2}{>{$}c<{$}}*{1}{>{$}l<{$}}}
    \otoprule
    \multicolumn{1}{b{1.2cm}}{\bfseries Related Work}
    & \multicolumn{1}{b{1.2cm}}{\bfseries Algorithm Type}
    & \multicolumn{1}{b{1.5cm}}{\bfseries Graph Type}
    & \multicolumn{1}{b{1.8cm}}{\bfseries Oblivious Function $f$}
    & \multicolumn{1}{b{1.5cm}}{\bfseries Oblivious Sources}
    & \multicolumn{1}{b{1.5cm}}{\bfseries Approx Factor} \\ 
    \otoprule 
    \raggedright Lujun Jia\\ \textit{et al.} \cite{JiaNRS06}
    & Deterministic
    & \raggedright Random\\Deployment
    & \times
    & \checkmark
    & O(\log n)\\
    \hline
    \multirow{2}*{}Lujun Jia
    & Deterministic  & Arbitrary Metric   & \times & \checkmark & O(\frac{\log^4 n}{\log \log(n)})\\
    \textit{et al.} \cite{1060649}
    & Deterministic  & Doubling Metric   & \times & \checkmark & O(\log(n))\\
    \hline
    \raggedright Ashish Goel\\ \textit{et al.} \cite{644191}
    & Randomized
    & General Graph $\bigtriangleup$-inequality
    & \checkmark
    & \times
    & O(\log k) \\
    \hline
    \raggedright Ashish Goel\\ \textit{et al.} \cite{10.1109/FOCS.2009.41}
    & Probabilistic
    & General Graph
    & \checkmark
    & \times
    & O(1) \\
    \hline
    \multirow{2}*{}Anupam Gupta  
    & Randomized  & General Graph   & \checkmark & \checkmark & O(\log^2 n)\\
    \textit{et al.} \cite{1109665}
    & Randomized & Low Doubling    & \checkmark & \checkmark  & O(\log n)\\
    \hline
    \rowcolor[gray]{.95}
    This paper
    & Deterministic
    & \raggedright Planar
    & \checkmark
    & \checkmark
    & O(\log n)\\
    \bottomrule

  \label{table:comparison}
  \end{tabular}
\end{table}

Goel \textit{et al.} in \cite{644191} build an overlay tree on a graph that satisfies triangle-inequality. Their technique is based on maximum matching algorithm that guarantees $(1 + \log k)$-approximation, where $k$ is the number of sources. Their solution is oblivious with respect to the fusion cost function $f$. In a related paper \cite{10.1109/FOCS.2009.41}, Goel \emph{et al.} construct (in polynomial time) a set of overlay trees from a given general graph such that the expected cost of a tree for any $f$ is within an $O(1)$-factor of the optimum cost for that $f$.

Jia \textit{et al.} in \cite{JiaNRS06} build a Group Independent Spanning Tree Algorithm (GIST) that constructs an overlay tree for randomly deployed nodes in an Euclidean 2 dimensional plane. The tree (that is oblivious to the number of data sources) simultaneously achieves $O(\log n)$-approximate fusion cost and $O(1)$-approximate delay. However, their solution assumes a constant fusion cost function. We summarize and compare the related work in Table \ref{table:comparison}.

Lujun Jia \textit{et al.} \cite{1060649} provide approximation algorithms for TSP, Steiner Tree and set cover problems. They present a polynomial-time $(O(\log(n)),O(\log(n)))$-partition scheme for general metric spaces. An improved partition scheme for doubling metric spaces is also presented that incorporates constant dimensional Euclidean spaces and growth-restricted metric spaces. The authors present a polynomial-time algorithm for Universal Steiner Tree (UST) that achieves polylogarithmic stretch with an approximation guarantee of $O(\log^4 n/\log \log(n))$ for arbitrary metrics and derive a logarithmic stretch, $O(\log(n))$ for any doubling, Euclidean, or growth-restricted metric space over $n$ vertices. They provide a lower bound of $\Omega(\log n/ \log \log n)$ for UST that holds even when all the vertices are on a plane.

Gupta \textit{et al.} \cite{1109665} develop a framework to model \textit{oblivious network design} problems (MSBB) and give algorithms with poly-logarithmic approximation ratio. They develop oblivious algorithms that approximately minimize the total cost of routing with the knowledge of aggregation function, the class of load on each edge and nothing else about the state of the network. Their results show that if the aggregation function is summation, their algorithm provides a $O(\log^2 n)$ approximation ratio and when the aggregation function is $max$, the approximation ratio is $O(\log^2 n \log \log n)$. The authors claim to provide a deterministic solution by derandomizing their approach. But, the complexity of this derandomizing process is unclear.

\vspace{-0.4cm}
\subsubsection{Non-Oblivious Network Design}
There has been a lot of research work in the area of approximation algorithms for network design. Since network design problems have several variants with several constraints, only a partial list has been mentioned in the following paragraphs.

The ``single-sink buy-at-bulk'' network design (SSBB) problem has a single ``destination'' node where all the demands from other nodes have to be routed to. Network design problems have been primarily considered in both Operations Research and Computer Science literatures in the context of flows with concave costs. The single-sink variant of the problem was first introduced by Salman \textit{et al.} \cite{589033}. They presented an $O(\log n)$-approximation for SSBB in Euclidean graphs by applying the method of Mansour and Peleg \cite{903727}. Bartal's tree embeddings \cite{bartal-tree} can be used to improve their ratio to $O(\log n \log\log n)$. A $O(\log^2 n)$-approximation was given by Awerbuch \textit{et al.} \cite{796341} for graphs with general metric spaces. Bartal \textit{et al.} \cite{276725} further improved this result to $O(\log n)$. Guha \cite{380827} provided the first constant-factor approximation to the problem, whose ratio was estimated to be around 9000 by Talwar \cite{660077}.

\subsubsection{Organization}
In the next section, we present some definitions and notations used throughout the rest of the paper. Section \ref{section:Covers} describes the concept of sparse covers and shortest path clustering. In addition, we derive a coloring for the cover. In section \ref{sec:Algorithm}, we describe {\sf{FixedPaths}} and {\sf{FindPaths}} algorithms that build a set of shortest paths between all pairs of nodes in $G$. Section \ref{section:Analysis} provides the analysis of the {\sf{FindPaths}} algorithm as well as the main theorem of this paper. Finally, we discuss our contribution and future work in section \ref{section:Conclusions}.


\section{Definitions}
\label{section:definitions}
Consider a weighted graph $G = (V,E,\ww)$,
where $\ww : E \rightarrow \mathbb{Z}^+$.
For any two nodes $u,v \in V$,
their {\em distance} $\dist(u,v)$ is the length of the shortest path
that connects the two nodes in $G$.
We denote by $N_k(v)$ the {\em $k$-neighborhood} of $v$
which is the set of nodes distance at most $k$ from $v$.
For any set of nodes $S \subseteq V$,
we denote by $N_k(S)$ the $k$-neighborhood of $S$
which contains all nodes which are within distance $k$ from
any node in $S$.

A set of nodes $X \subseteq V$
is called a {\em cluster} if the induced subgraph $G(X)$ is connected.
Let $Z = \{X_1, X_2, \ldots, X_k \}$ be a set of clusters in $G$.
For every node $v \in G$, let $Z(v) \subseteq Z$ denote the set
of clusters that contain $v$.
The {\em degree} of $v$ in $Z$ is defined as $\beta_v(Z) = |Z(v)|$,
which is the number of clusters that contain $v$.
The degree of $Z$ is defined as
$\beta(Z) = \max_{v \in V} \beta_v(Z)$,
which is largest degree of any of its nodes.
The radius of $Z$ is defined as $\radius(Z) = max_{X \in Z} (\radius(X))$.

Consider a {\em locality parameter} $\gamma > 0$.
A set of clusters $Z$ is said to
{\em $\gamma$-satisfy} a node $v$ in $G$, if there is a cluster $X \in Z$,
such that the $\gamma$-neighborhood of $v$, $N_{\gamma}(v)$,
(nodes within distance $\gamma$ from v) is included in $X$,
that is, $N_{\gamma}(v) \subseteq X$.
A set of clusters $Z$ is said to be a {\em $\gamma$-cover} for $G$,
if every node of $G$ is $\gamma$-satisfied by $Z$ in $G$.
The {\em stretch} $\sigma(Z)$ of a $\gamma$-cover $Z$ is the smallest number
such that $\radius(Z) = \sigma(Z) \cdot \gamma$.

We define the following coloring problem in a set of clusters $Z$.
We first define the notion of the distance between
two clusters $X_i, X_j \in Z$, $X_i \neq X_j$.
We say that
$\dist(X_i, X_j) \leq k$,
if there is a pair of nodes $u \in X_i$ and $v \in X_j$
such that $u$ is $k$-satisfied in $X_i$, $v$ is $k$-satisfied
in $X_j$, and $\dist(u,v) \leq k$.
A valid distance-$k$ coloring of $Z$ with a palette of $\chi$ colors $[1,\chi]$,
is an assignment of an integer $\col(X) \in [1,\chi]$ to every $X \in Z$,
such that there is no pair of clusters $X_i, X_j \in Z$, $X_i \neq X_j$,
with $\dist(X_i, X_j) \leq k$ which receive the same color.
The objective is to find the smallest $\chi$ that permits a valid distance-$k$ coloring.

\section{Sparse Cover}
\label{section:Covers}

A $\gamma$-cover is {\em sparse} if it has small degree and stretch.
In \cite[Section 5]{1281112}
the authors present a polynomial time sparse cover construction
algorithm {\sf Planar-Cover}$(G,\gamma)$
for any planar graph $G$ and locality parameter $\gamma$,
which finds a $\gamma$-cover $Z$
with constant degree, $\beta \leq 18$,
and constant stretch, $\sigma \leq 24$.
Here, we show that this cover also
admits a valid distance-$\gamma$
coloring with a constant number of colors $\chi \leq 18$.

For any node $v \in G$, we denote by $\depth_v(G)$ the shortest
distance between $v$ and an external node (in the external face) of $G$.
We also define $\depth(G) = \max_{v \in V} \depth_v(G)$.
The heart of sparse cover algorithm in \cite[Section 5]{1281112}
concerns the case where $\depth(G) \leq \gamma$
which is handled in Algorithm {\sf Depth-Cover}$(G, \gamma)$.
The general case, $\depth(G) > \gamma$,
is handled by diving the graph into zones of depth $O(\gamma)$,
as we discuss later.
So, assume for now that $\depth(G) \leq \gamma$.

The Algorithm {\sf Depth-Cover}$(G, \gamma)$,
relies on forming clusters along shortest paths
connecting external nodes (in the external face) of $G$.
For every shortest path $p$,
Algorithm {\sf Shortest-Path-Cluster}$(G,p,4\gamma)$ in \cite[Section 3]{1281112}
returns a set of clusters around the $4\gamma$ neighborhood of $p$
with radius at most $8\gamma$ and degree $3$.
Then, $p$ and all its $2\gamma$-neighborhood is removed from $G$
producing a smaller subgraph $G'$ (with possibly multiple connected components).
The algorithm proceeds recursively on each connected component $H$ of $G'$
by selecting an appropriate new shortest path $p'$ between external nodes of $H$.
The algorithm terminates when all the nodes have been removed.
The initial shortest path that starts the algorithm
consists of a single external node in $G$.
The resulting $\gamma$-cover $Z$
consists of the union of all the clusters from all the shortest paths.
The shortest paths are chosen in such a way that a node participates
in the clustering process of at most 2 paths, and this bounds the degree of
the $\gamma$-cover to be at most $\beta \leq 6$, and stretch $s \leq 8$.

The analysis in \cite[Section 5.1.1]{1281112}
of Algorithm {\sf Depth-Cover}
relies on representing the clustering process of $G$ as a tree $T$
as we outline here.
Each tree node $w \in T$ represents a pair $w = (G(w),p(w))$
where $G(w)$ is a planar subgraph of $G$ that is to be clustered,
and $p(w)$ is a shortest path between two external nodes of $G(w)$.
The root of the tree is $r = (G, v)$,
where $v$ is a trivial initial path with one external node $v \in G$.
The children of a node $w \in T$ are all the nodes $w' = (G(w'),p(w'))$,
such that $G(w')$ is a connected component that results from
removing $p(w)$ and its $2\gamma$-neighborhood from $G$.

Next, we extend \cite[Lemma 3.1]{1281112}
to show that we can color the clusters obtained by
a shortest path clustering using a constant number of colors. The proof of the following three lemmas is given in the appendix.

\begin{lemma}
\label{lemma:color-path}
For any graph $G$, shortest path $p \in G$,
the set of clusters returned by
Algorithm {\sf Shortest-Path-Cluster($G,p, 4\gamma$)}
admits a valid distance-$\gamma$ coloring with $3$ colors.
\end{lemma}

We can obtain a coloring of $Z$ by coloring
the respective levels of the tree $T$.
Assume that the root is at level $0$.

\begin{lemma}
\label{lemma:color-level}
The union of clusters in any level $i \geq 0$ of the tree $T$,
admits a valid distance-$\gamma$ coloring with 3 colors.
\end{lemma}

\begin{lemma}
\label{lemma:color-depth}
Algorithm {\sf{Depth-Cover}}$(G, \gamma)$
returns a set of clusters $Z$ which
admits a valid distance-$\gamma$ coloring with 6 colors.
\end{lemma}

We are now ready to consider the case $\depth(G) > \gamma$.
Algorithm {\sf Planar-Cover}$(G, \gamma)$
decomposes the graph into a sequence of bands,
such that each band $W_i$ has depth $\gamma$.
The bands are organized into zones, such that zone $S_i$
consists of three consecutive bands $W_{i-1}, W_i, W_{i+1}$.
Thus, zone $S_i$ overlaps with bands $S_{i-2}$, $S_{i-1}$, $S_{i+1}$ and $S_{i+2}$.
The algorithm invokes {\sf Depth-Cover}$(S_i, 3\gamma-1)$ for each zone
giving a $\gamma$-cover $Z$ with degree $\beta \leq 3 \cdot 6 = 18$ and
stretch $\sigma \leq 3 \cdot 8 = 24$.

We can obtain the following coloring result.
Using Lemma \ref{lemma:color-depth},
for every zone $S_i$ we can get a valid distance-$(3\gamma-1)$ coloring
with a palette of $6$ colors.
This implies that we can obtain a valid distance-$\gamma$ coloring
for the zone with at most $6$ colors.
Zones $S_i$ and $S_{i+3}$ do not overlap
and any two nodes satisfied in them (one from each zone)
with respect to $G$
must be more than $\gamma$ distance apart.
Therefore,
we can color all the zones with three different palettes each consisting of 6 colors,
so that zone $S_{i}$,
uses the $((i \mod 3) + 1)$th palette.
The coloring can be found in polynomial time.
Therefore, we obtain:

\begin{theorem}
\label{theorem:color-main}
Algorithm {\sf{Planar-Cover}}$(G, \gamma)$
produces a set of clusters $Z$ which
has degree $\beta = 18$, stretch $\sigma \leq 24$,
and admits a valid distance-$\gamma$ coloring with $\chi = 18$ colors.
\end{theorem}

\section{Algorithm}
\label{sec:Algorithm}

We describe how to find the shortest path between a given pair of nodes in a graph $G = (V, E, \ww)$.
To find such paths, we use Algorithm {\sf FindPaths}
(Algorithm \ref{algo:FindPaths})
which relies on Algorithm {\sf AuxiliaryPaths}
(Algorithm \ref{algo:AuxiliaryPaths}).


\begin{algorithm}[t]
\KwIn{Graph $G = (V,E,\ww)$}
\KwOut{Set of auxiliary paths for all nodes in $G$}
\BlankLine

\tcp{$Z_i$ is a $\gamma_i$-cover of $G$, $0 \leq i \leq \kappa$
($\gamma_i$, $\kappa$ specified in Section \ref{section:Analysis})}
\tcp{Assume each cluster has a designated leader node}
\BlankLine

$\mathcal{Q} \gets \emptyset$; ~~~\tcp{set of auxiliary paths for all nodes in $G$}

\ForEach{$v \in V$}{
    $q(v) \gets \emptyset$; ~~~\tcp{auxiliary path for node $v$ from level $0$ to $\kappa$}
	$x \gets v$\;
	\For {$i = 0$ \emph{\KwTo}$\kappa-1$}{
		Let $X \in Z_{i+1}$ be a cluster that $\gamma_{i+1}$-satisfies $v$\;
		$\ell_{i+1}(v) \gets$ leader of $X$\;
		$q_i(v) \gets$ shortest path from $x$ to $\ell_{i+1}(v)$\;
		$q(v) \gets$ concatenate $q(v)$ and $q_i(v)$\;
		$x \gets \ell_{i+1}(v)$\;
	}
	$\mathcal{Q} \gets \mathcal{Q} \cup q(v)$\;
}

\Return $\mathcal{Q}$\;
\caption{{\sf AuxiliaryPaths}$(G)$}
\label{algo:AuxiliaryPaths}
\end{algorithm}


\begin{algorithm}
\KwIn{Graph $G = (V,E,\ww)$.}
\KwOut{Set of paths between all pair of nodes in $G$}
\BlankLine

\tcp{Let covers $Z_i$ be as in Algorithm {\sf AuxiliaryPaths}}
\tcp{Let leader $\ell_i(x)$ be as defined in Algorithm {\sf AuxiliaryPaths}; $\ell_0(x) \gets x$}
\tcp{Let auxiliary path $q(x)$ be as computed in Algorithm {\sf AuxiliaryPaths}}
\tcp{Let ${\overline q}_i(x)$ be the auxiliary path segment from $x$ to $\ell_i(x)$; ${\overline q}_{0}(x) \gets x$}
\BlankLine

$\mathcal{P} \gets \emptyset$; ~~~\tcp{set of paths for all pairs of nodes in $G$}
\ForEach{pair $u,v \in V$} {
    Let $\gamma_i$ be the smallest locality parameter such that $\gamma_i \geq 2\dist(u,v)$\;
	Let $X \in Z_{i}$ be a cluster that $\gamma_{i}$-satisfies $u$ (and hence $\gamma_i/2$-satisfies $v$)\;
    Let $\ell'$ be the leader of $X$ ({\em common leader} of $u,v$)\;
	$q' \gets$ concatenate shortest paths from $\ell_{i-1}(u)$ to $\ell'$ to $\ell_{i-1}(v)$\;
	$p(u,v) \gets$	concatenate ${\overline q}_{i-1}(u)$, $q'$, and ${\overline q}_{i-1}(v)$\;
	$\mathcal{P} \gets \mathcal{P} \cup p(u,v)$\;
}
		
\Return $\mathcal{P}$\;

\caption{{\sf FindPaths}$(G)$}
\label{algo:FindPaths}
\end{algorithm}


Both algorithms use $\kappa+1$ covers $Z_0, \ldots, Z_{\kappa}$,
where in $Z_0$ every node in $V$ is a cluster,
and $Z_i$ is a $\gamma_i$-cover of $G$, for $i \geq 1$,
where the parameters $\gamma_i$ and $\kappa$ are defined in Section \ref{section:Analysis}.
We refer to the cover $Z_i$ as the level $i$ cover of $G$.
We assume that each cluster in the covers has a designated leader node.
There is a unique cluster, containing all nodes in $G$,
and leader node $\ell_{\kappa}$ at level $\kappa$.

Algorithm {\sf AuxiliaryPaths} computes a {\em auxiliary path} $q(v)$
from every node $v \in V$ to $\ell_{\kappa}$.
The auxiliary paths are built in a bottom-up fashion.
A auxiliary path from any node $v \in V$ at level $0$ is built recursively.
In the basis of the recursion,
we identify a cluster $X_1 \in Z_1$,
which $\gamma_1$-satisfies node $v$.
Let $\ell_1(v)$ denote the leader $X_1$.
We now compute a shortest path,
denoted $q_0(v)$, from $v$ to $\ell_1(v)$.
This forms the first path segment of $q(v)$.
Suppose we have computed $q(v)$ up to level $i$, $i < \kappa$.
We now want to extend this path to the next higher level $i+1$.
To compute the path segment from level $i$ to level $i+1$,
we repeat the process of finding a cluster $X_{i+1} \in Z_{i+1}$
that $\gamma_{i+1}$-satisfies node $v$.
Let $\ell_{i+1}(v)$ denote the leader $X_1$.
We compute the shortest path,
denoted $q_i(v)$ from $\ell_i(v)$
to $\ell_{i+1}(v)$.
We then append this new path segment $q_i(v)$ to $q(v)$
to form the current extended path $q(v)$.
The path building process terminates when the last leader reaches level $\kappa$.

We are now ready to describe how Algorithm {\sf FindPath}
computes the shortest paths between all pair of nodes in $G$.
For pair of nodes $u,v \in V$,
let $y$ be the distance between them.
Let $\gamma_i$ be the smallest locality parameter
such that $\gamma_i \geq y/2$.
Let $X \in Z_i$ be the cluster that $\gamma_i$ satisfies $u$,
and let $\ell'$ be the respective leader of $X$.
Note that by the way that we have chosen $\gamma_i$,
cluster $X$ also $\gamma_i/2$-satisfies $v$.
Let ${\overline q}_i(u)$ denote the segment of the auxiliary path $q(u)$
from $u$ to $\ell_i(u)$.
We concatenate ${\overline q}_{i-1}(u)$,
with a shortest path from $\ell_{i-1}(u)$ to $\ell'$,
with a shortest path from $\ell'$ to $\ell_{i-1}(v)$,
and ${\overline q}_{i-1}(v)$.
This gives the path $p(u,v)$.

\section{Analysis}
\label{section:Analysis}
Let $G = (V,E,\ww)$ be a planar graph with $n$ nodes.
In this section we use the following parameters:
\begin{center}
\begin{tabular}{ll}
$\kappa  =  1 + \lceil \log_{4\sigma} D \rceil$ & //highest cluster level in $G$\\
$\beta   =  18$ & //cover degree bound\\
$\sigma  =  24$ & //cover stretch bound \\
$\gamma_i  =  (4\sigma)^{i-1}$ & //locality parameter of level $i \geq 1$ cover \\
$\chi  =   18$ & //coloring of level $i$
\end{tabular}
\end{center}
Consider $\kappa + 1$ levels of covers $Z_0, \ldots, Z_{\kappa+1}$,
where in $Z_0$ each node in $V$ is a cluster, and each $Z_i$, $i \geq 1$,
is a $\gamma_i$-cover of $G$
which is obtained from Theorem \ref{theorem:color-main}.
Thus, each $Z_i$, $i\geq 1$, has degree at most $\beta$, stretch at most $\sigma$,
and can be given a valid distance-$\gamma_i$ coloring with $\chi$ colors.

Let $A$ denote an arbitrary set of demands.
For any demand $d = (s,t) \in A$
let $p(d) = p(s,t)$ be the path given by Algorithm {\sf FindPaths}.
Suppose that the common leader of $s$ and $t$ is $\ell$.
The path $p(d)$ consists of two path segments:
the {\em source path segment $p(s)$}, from $s$ to $\ell$,
and the {\em destination path segment $p(t)$} from $\ell$ to $t$.
We denote by $p_i(s)$ the subpath between level $i$ and level $i+1$
(we call this the level $i$ subpath).

Let $C^*(A)$ denote the cost of optimal paths in $A$.
Let $C(A)$ denote the cost of the paths given by our algorithm.
We will bound the competitive ratio $C(A)/C^*(A)$.
For simplicity, in the approximation analysis,
we consider only the cost of the source path segments $p(s_i)$.
When we consider the destination segments the
approximation ratio increases by a factor of 2.

The cost $C(A)$ can be bounded as a summation
of costs from the different levels as follows.
For any edge $e$ let
$\varphi_{e,i}(A) = \{ p_i(s) : ((s,t) \in A) \wedge (e \in p_i(v)) \}$
be the set of layer-$i$ subpaths that use edge $e$.
Denote by $C_{e,i}(A) = f(|\varphi_{e,i}(A)|) \cdot w_e$
the cost on the edge $e$ incurred by the level-$i$ subpaths.
Since $f$ is subadditive,
we get $C_e(A) \leq \sum_{i=0}^{\kappa-1} C_{e,i}(A)$.
Let $C_i(A) = \sum_{e \in E} C_{e,i}(A)$ denote the
cost incurred by the layer-$i$ subpaths.
Since $C(A) = \sum_{e \in E} C_e(A)$,
we have that:
\begin{equation}
\label{eqn:cost-levels}
C(A) \leq  \sum_{i=0}^{\kappa-1} C_i(A).
\end{equation}

For any cluster $X$ let $X(A)$
denote the set of demands with source in $X$
whose paths leave from the leader of $X$
toward the leader of a higher level cluster.

\begin{lemma}
\label{lemma:lower-bound-big}
For any $Z_i$, $2 \leq i \leq \kappa-1$,
$C^*(A) \geq R(i) / \chi$,
where $R(i) = \sum_{X \in Z_i} f(|X(A)|) \cdot \gamma_{i}/2$.
\end{lemma}

\begin{proof}
Let $Z_i(k)$ to be the set of clusters at level $i$
which receive color $k \in [1,\chi]$.
Consider a cluster $X \in Z_i(k)$.
Consider a demand $(s,t) \in X(A)$.
Since $X\in Z_i(k)$ the common leader of $s$ and $t$ is at a level $i+1$ or higher.
From the algorithm, $\dist(s,t) \geq \gamma_{i+1}/2$.
Consider the subpaths from $X(A)$ of length up to $\gamma_{i}/2$.
In the best case,
these subpaths from $X(A)$ may be combined
to produce a path with smallest possible total cost $f(|X(A)|) \cdot \gamma_{i}/2$.
Any two nodes $u \in X(A)$ and $v \in X(Y)$,
where $X, Y \in Z_{i}(k)$ and $X \neq Y$,
have $\dist(u,v) > \gamma_i$,
since each node is $\gamma_i$-satisfied in its respective cluster
and $X$ and $Y$ receive the same color in the
distance-$\gamma_i$ coloring of $Z$.
Therefore, the subpaths of lengths up to $\gamma_{i}/2$
from the demands $X(A)$ and $Y(A)$ cannot combine.
Consequently,
$C^*(A) \geq R(i,k)$
where
$R(i,k) = \sum_{X \in Z_i(k)} f(|X(A)|) \cdot \gamma_{i}/2$.
Let $R_{\max} = \max_{k \in [1, \chi]} R(i,k)$.
We have that $C^*(A) \geq R_{\max}$.
Since $R(i) = \sum_{k = 1}^{\chi} R(i,k) \leq R_{\max} \cdot \chi$.
We obtain $C^*(A) \geq R(i) / \chi$,
as needed.
\end{proof}

We also get the following trivial lower bound for the special
case where $0 \leq i \leq 1$,
which follows directly from
the observation that each demand needs to form a path with length at least 1.

\begin{lemma}
\label{lemma:lower-bound-small}
For any $Z_i$, $0 \leq i \leq 1$,
$C^*(A) \geq \sum_{X \in Z_i} f(|X(A)|)$.
\end{lemma}

We obtain the following upper bound.

\begin{lemma}
\label{lemma:upper-bound}
For any $Z_i$, $0 \leq i \leq \kappa-1$,
$C_i(A) \leq Q(i)$
where $Q(i) = \sum_{X \in Z_i} f(|X(A)|) \cdot \beta \sigma \gamma_{i+1}$.
\end{lemma}

\begin{proof}
For any cluster $X \in Z_i$,
we can partition the demands $X(A) = Y_1 \cup Y_2 \cup \ldots \cup Y_k$,
where $Y_i \neq Y_j$, $i \neq j$,
according to leader at level $i+1$ that they use,
so that all demands in $Y_i$ use the same leader in $Z_{i+1}$,
and $Y_i$ and $Y_j$ use a different leader of $i+1$.
Next, we provide a bound on $k$.

Consider any two demands $d_1 = (s_1, t_1) \in X(A)$ and $d_2 = (s_2, t_2) \in X(A)$.
Let $\ell_i$ be the common leader at level $i$.
Since $s_1$ and $s_2$ are $\gamma_i$-satisfied
by the cluster of $\ell_i$,
they are both members of that cluster.
Therefore, $\dist(s_1, \ell_i) \leq \sigma \gamma_i$,
and $\dist(s_2, ell_i) \leq \sigma \gamma_i$.
Thus, $\dist(s_1, s_2) \leq 2\sigma \gamma_i = \gamma_{i+1}/2$.
Suppose that demand $d_1$ chooses leader $\ell_{i+1}$ at level $i+1$ with respective cluster $X_{i+1}$.
Since $s_1$ is at least $\gamma_{i+1}/2$-satisfied in $X_{i+1}$,
$s_2$ is a member of $X_{i+1}$.
Since any node is a member of at most $\beta$ clusters at level $i+1$,
it has to be that the number of different level $i+1$ leaders
at level $i+1$ that the demands in $X(A)$ select is bounded by $\beta$.
Consequently, $k \leq \beta$.

Since $f$ is subadditive and for any demand $(s,t)$
$|p_i(s)| \leq \sigma \gamma_{i+1}$,
$C_i(Y_j) \leq f(|Y_j|) \cdot \sigma \gamma_{i+1}$.
Therefore,
$C_i(X(A)) \leq \sum_{j = 1}^{k} C_i(Y_j) \leq f(|X(A)|) \cdot \beta \sigma \gamma_{i+1}$.
Which gives:
$C_i(A) \leq \sum_{X \in Z_i} f(|X(A)|) \cdot \beta \sigma \gamma_{i+1}$,
as needed.
\end{proof}

\begin{lemma}
\label{lemma:second-upper-bound}
For any $0 \leq i \leq \kappa-1$,
$C_i(A) \leq  C^*(A) \cdot 8 \beta \sigma^2 \chi$.
\end{lemma}

\begin{proof}
From Lemma \ref{lemma:upper-bound},
for any $0 \leq i \leq \kappa-1$,
$C_i(A) \leq Q(i)$.
From Lemma \ref{lemma:lower-bound-big},
for any $2 \leq i \leq \kappa - 1$,
$C^*(A) \geq R(i) / \chi$.
Note that
$Q(i)
= R(i) \cdot 2 \beta \sigma \chi \gamma_{i+1} / \gamma_{i}
= R(i) \cdot 8 \beta \sigma^2 \chi$.
Therefore,
$C_i(A)
\leq C^*(A) \cdot 8 \beta \sigma^2 \chi$.
For $0 \leq i \leq 1$,
we use the lower bound of Lemma \ref{lemma:lower-bound-small},
and we obtain
$C_i(A) \leq  C^*(A) \cdot \beta \sigma \gamma_2 = C^*(A) \cdot 4 \beta \sigma^2$.
\end{proof}

We now give the central result of the analysis:

\begin{theorem}
\label{theorem:main}
The oblivious approximation ratio of the algorithm is $O(\log n)$.
\end{theorem}

\begin{proof}
Since the demand set $A$ is arbitrary,
from Lemma \ref{lemma:second-upper-bound}
and Equation \ref{eqn:cost-levels}
we obtain oblivious approximation ratio bounded by
$8 \kappa \beta \sigma^2 \chi$.
When
we take into account the source path segments together
with the destination path segments,
the approximation ratio bound increases by a factor of 2,
and it becomes $16 \kappa \beta \sigma^2 \chi$.
Since, $\beta$, $\sigma$, $\chi$, are constants and $\kappa = O(\log n)$,
we obtain approximation ratio $O(\log n)$.
\end{proof}

\section{Conclusions}
\label{section:Conclusions}

We provide a set of paths for the multi-sink buy-at-bulk network design problem in planar graphs. 
Contrary to many related work where the source-destination pairs were already given, 
or when the source-set was given, we assumed the obliviousness of the set of source-destination pairs. 
Moreover, we considered an unknown fusion cost function at every edge of the graph. 
We presented nontrivial upper and lower bounds for the cost of the set of paths. 
We have demonstrated that a simple, deterministic, 
polynomial-time algorithm based on sparse covers and shortest-path clustering 
can provide a set of paths between all pairs of nodes in $G$
that can accommodate any set of demands. 
We have shown that this algorithm guarantees $O(\log n)$-approximation. 
As part of our future work, we are looking into obtaining efficient solutions 
to other related network topologies, such as minor-free graphs.

\bibliographystyle{plain}
\bibliography{Bibliography}



\appendix
\newpage
\pagenumbering{roman}
\section*{Appendix: Proofs}

\paragraph{Proof of Lemma \ref{lemma:color-path}}
The algorithm divides the path $p$ into consecutive disjoint subpaths
$p_1, p_2, \ldots, p_\ell$
each of length $4 \gamma$
(except for the last subpath $p_\ell$ which may have shorter length).
The algorithm builds a cluster $X_i$ around each subpath $p_i$
which consists of the $4 \gamma$-neighborhood of $p_i$.
We can show that $\dist(X_i, X_{i+3}) > \gamma$.
Suppose otherwise.
Then, there are nodes $u \in X_i$, $v \in X_{i+3}$,
which are $\gamma$-satisfied in their respective clusters
and $\dist(u,v) \leq \gamma$.
Thus, $u \in X_{i+3}$.
Then, there is a path of length at most $8 \gamma$ that connects
the two paths $p_i$ and $p_{i+3}$ which is formed through $u$.
However, this is impossible since the paths
are at distance at least $8 \gamma + 1$.
Therefore,
we can use a palette of at most $3$ colors to color the
clusters, so that each cluster $X_i$ receives color $(i \mod 3) + 1$.

\paragraph{Proof of Lemma \ref{lemma:color-level}}
Consider a level $i \geq 0$ of $T$.
From  Lemma \ref{lemma:color-path},
the clusters produced in any node $w$ of level $i$
from path $p(w)$
can be colored with $3$ colors.
Consider now two nodes $w_1$ and $w_2$ in level $i$ of $T$.
Let $X$ be a cluster from $p(w_1)$.
Any node which is $\gamma$-satisfied (with respect to $G$)
in $X$, cannot have in its $\gamma$-neighborhood
any node in $G(w_2)$,
since $G(w_1)$ and $G(w_2)$ are disjoint.
Therefore, any two nodes which are $\gamma$-satisfied in
the respective clusters of $w_1$ and $w_2$ have to be at distance
more than $\gamma$ from each other in $G$.
This implies that we can use same palette of 3 colors
for each node in the same level $i$ of the tree.

\paragraph{Proof of Lemma \ref{lemma:color-depth}}
From Lemma \ref{lemma:color-level},
the clusters of each level of the tree $T$ can be
colored with 3 colors.
From the proof in \cite[Lemma 5.4]{1281112},
any node $v \in G$ is clustered in at most $2$ consecutive levels $i, i+1$
of $T$, and does not appear in any subsequent level.
Any node $u$ which is $\gamma$-satisfied in a cluster of level $i+2$ cannot
be with distance of $\gamma$ or less from $v$,
since $v$ doesn't appear in the level $i+2$ subgraph of $G$.
Therefore,
any node $v$ which is $\gamma$-satisfied in a cluster of level $i$
must be at distance more than $\gamma$ than any node $u$ which is $\gamma$-satisfied
in a cluster of level $i+2$.
Therefore the clusters formed at level $i$ are at distance at least $\gamma+1$
from clusters formed at level $i+2$.
Consequently, we can use color palette $[1,3]$ for odd levels
and color palette $[4, 6]$ for even levels,
using in total $6$ colors.

\end{document}